\documentclass[aps,twocolumn,floatfix]{revtex4-2}

\usepackage{amsmath,amsthm,amsfonts,amssymb,times,bbm,graphicx,bbold,mathrsfs}
\usepackage[hidelinks, colorlinks=true, allcolors=blue]{hyperref}
\usepackage{float}
\graphicspath{{./images/}}

\newtheorem{theorem}{Theorem}

\newtheorem{lemma}[theorem]{Lemma}

\newtheorem{definition}[theorem]{Definition}

\newtheorem{remark}[theorem]{Remark}

\newcommand{\RR}{\mathbb{R}}
\newcommand{\ZZ}{\mathbb{Z}}
\newcommand{\NN}{\mathbb{N}}

\renewcommand{\S}{\mathcal{S}}
\newcommand{\E}{\mathcal{E}}
\renewcommand{\P}{\mathcal{P}}
\newcommand{\D}{\mathcal{D}}

\newcommand{\1}{\underline{1}}
\newcommand{\Id}{\mathbb{1}}

\DeclareMathOperator{\tr}{tr}

\usepackage{xcolor}

\usepackage{algorithm}
\usepackage{algorithmicx}
\usepackage{algpseudocode}
\algrenewcommand\algorithmiccomment[1]{\hfill {#1}}

\usepackage{stmaryrd}
\newcommand{\swap}[3]{\left\llbracket {}_{#1}{}^{#2}{}_{#3} \right\rrbracket}
\newcommand{\dswap}[3]{\left\llbracket {}^{#1}{}_{#2}{}^{#3} \right\rrbracket}

\begin{document}

\title{Entanglement-swapping in generalised probabilistic theories, and iterated CHSH games}

\author{Lionel J.\ Dmello}
\email{ldmello@thp.uni-koeln.de} 
\author{Laurens T.\ Ligthart}
\email{ligthart@thp.uni-koeln.de}
\author{David Gross} 
\email{david.gross@thp.uni-koeln.de}
\affiliation{ Institute for Theoretical Physics, University of Cologne, Germany } 
\date{May 22, 2024} 

\begin{abstract} 
    While there exist theories that have states ``more strongly entangled'' than quantum theory, in the sense that they show CHSH values above Tsirelson's bound, all known examples of such theories have a strictly smaller set of measurements.
    Therefore, in tasks which require both bipartite states and measurements, they do not perform better than QM.
	One of the simplest information processing tasks involving both bipartite states and measurements is that of \emph{entanglement swapping}.
	In this paper, we study entanglement swapping in \emph{generalised probabilistic theories} (GPTs).
	In particular, we introduce the \emph{iterated CHSH game}, which measures the power of a GPT to preserve non-classical correlations, in terms of the largest CHSH value obtainable after $n$ rounds of entanglement swapping.
	Our main result is the construction of a GPT that achieves a CHSH value of $4$ after an arbitrary number of rounds.
	This addresses a question about the optimality of quantum theory for such games recently raised by Weilenmann and Colbeck. 
    One challenge faced when treating this problem is that there seems to be no general framework for constructing GPTs in which entanglement swapping is a well-defined operation.
    Therefore, we introduce an algorithmic construction that turns a bipartite GPT into a multipartite GPT that supports entanglement swapping, if consistently possible.
\end{abstract} 

\maketitle

\section{Introduction}\label{sec:intro}

Finding a set of operationally motivated axioms that single out QM has been a long-standing problem in the field of foundations of quantum mechanics.
Such an undertaking requires a mathematical framework that allows us to compare QM with other theories, such as classical theory.
Arguably the most general framework is that of \emph{generalised probabilistic theories} (GPTs) \cite{segal1947postulates,ludwig1967attempt2,ludwig1968attempt3,dahn1968attempt,stolz1969attempt,stolz1971attempt,davies1970operational,barrett2007information, plavala2023general}.

Within this framework, a large number of axioms have been proposed over the years \cite{hardy2001quantum, pawlowski2015information, miklin2021information, brassard2006limit, linden2007quantum, navascues2010glance, fritz2013local}.
One of the main foci of many investigations is the behavior of bipartite theories.
In particular, explaining why CHSH-type experiments \cite{clauser1969proposed, bell1964einstein} in nature are bounded by $2\sqrt2$, which Tsirelson \cite{cirel1980quantum} famously showed to be the largest value allowed by quantum theory.
But, as it stands, we still have no definitive axiom singling out quantum theory based on this property.

In their recent work \cite{weilenmann2020toward}, Weilenmann and Colbeck 
turn to multipartite theories to possibly explain 
$2\sqrt2$.
The argument hinges on a well-known tension between the set of states and measurements in a theory:
Expanding state space to include more correlations requires one to shrink the set of measurements in order to avoid the emergence of negative probabilities.

For example, bipartite boxworld \cite{popescu1994quantum, barrett2007information, plavala2023general} state space is the biggest possible bipartite state space one could make. 
But, as a consequence, the bipartite effect space (the space from which the measurements are constructed) has only product effects. 
This in turn, is the smallest possible bipartite effect space.

Therefore, by adding a step of \emph{entanglement swapping} before playing the CHSH game, the set of possible correlations achievable in theories with such state spaces shrink (see~\cite{short2006entanglement, skrzypczyk2009emergence}). 
Quantum theory seems to strike the optimum between these competing notions, in that
the CHSH violation of $2\sqrt2$ is preserved under entanglement swapping.
And indeed, Weilenmann and Colbeck show that, for the \emph{adaptive CHSH game} \cite{weilenmann2020self}, the winning probability of any theory whose unipartite state and effect spaces are characterised by regular polygons \cite{janotta2011limits} is upper bounded by the winning probability of quantum theory.

Naturally, the question arises whether quantum theory is optimal, in the sense that no post-quantum theory can sustain a CHSH violation greater than $2\sqrt2$ through entanglement swapping.
To investigate this, we introduce the \emph{iterated CHSH game}, which is an extension of the adaptive CHSH game to multiple rounds.
The main result of this paper is to answer this question in the negative, with the construction of a post-quantum GPT which 
sustains the maximal-possible value of $4$ over arbitrarily many rounds in the iterated CHSH game.
Along the way, we also discuss a simpler construction which sustains a CHSH violation of $4$ only for a finite number of rounds.


\subsection{Generalised Probabilistic Theories}

The main idea of GPTs is to model experiments as a two step process. A \emph{preparation} step which produces a \emph{state}, and a \emph{measurement} step which probabilistically maps states to outcomes. 
Associated to each outcome of a measurement, comes an \emph{effect}. 

Mathematically, states are modelled as members of a convex set (convexity corresponding operationally to probabilistic mixtures).
Effects can be viewed as positive linear functionals on states.
The pairing between a state $\rho$ and an effect $e$ is interpreted as the probability of obtaining the outcome corresponding to the effect $e$, given that we prepared the state $\rho$.

It is useful to represent these operations diagrammatically. 
For example, the CHSH experiment consists of a bipartite state measured locally by two parties. 
The resulting contraction can be diagrammatically represented as in Fig.~\ref{fig:CHSH}.
This diagram corresponds to the joint probability obtained by contracting the $2$-tensor representing the state with the tensor product of the local unipartite effects, i.e., the pairing 
\begin{equation}
    \langle \rho_{12}, e_{1} \otimes f_{2} \rangle
\end{equation}
(see also Sec.~\ref{sec:theory_axioms}).

\begin{figure}
    \centering
    \includegraphics[width=0.2\textwidth, height=25mm, keepaspectratio]{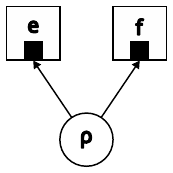}
    \caption{Diagramatic representation of the measurement of a bipartite state by two local observers. 
		Here, the bipartite state $\rho$ is represented by a circle with two arrows emerging from it, one for each sub-system. 
	The unipartite effects $e$ and $f$ on the other hand, are represented by boxes with one port each. The full diagram represents the probability that, on measuring $\rho$, Alice and Bob get the outcome corresponding to effect $e$ and $f$ respectively.}
    \label{fig:CHSH}
\end{figure}


\subsection{Entanglement swapping}

Contractions of the type depicted in Fig.~\ref{fig:CHSH}, whose result is a probability, are called \emph{full} contractions.
It is also possible to define \emph{partial contraction} of states and effects. 

Note for example the situation in Fig.~\ref{fig:partial_contractions}~(a), where a bipartite state $\rho$ is measured only on sub-system $1$ and the outcome corresponding to the effect $e$ is obtained.
We will interpret the resulting object as a (not necessarily normalised) \emph{conditional state}, which acts on effects $f$ on the second sub-system as
\begin{equation}
    \langle \rho_{12}, e_1 \otimes \Id_2 \rangle:
		f_2 \mapsto
    \langle \rho_{12}, e_1 \otimes f_2 \rangle.
\end{equation}
Here $\Id$ is the identity channel, which operationally corresponds to ``do nothing''.

Similarly, one can pair a bipartite effect with a unipartite state.
This situation is depicted in Fig.~\ref{fig:partial_contractions}~(b) and mathematically corresponds to
\begin{equation}
    \langle \rho_1 \otimes \Id_2, e_{12} \rangle
    :
    \sigma_2 \mapsto \langle \rho_1 \otimes \sigma_2, e_{12} \rangle.
\end{equation}

We can also extend these notions to the case of several bipartite objects.
Consider the situation where two bipartite states are partially contracted with a bipartite effect, as depicted in Fig.~\ref{fig:partial_contractions}~(c).
This is the generalisation of the quantum-mechanical notion of entanglement swapping to GPTs. 
These operations will occur so frequently, that we introduce a compact notation:
For bipartite states $\rho, \sigma$ and a bipartite effect $e$, we denote 
the conditional bipartite state resulting from an entanglement swapping procedure as
\begin{equation}
    \swap{\rho}{e}{\sigma}:= \langle \rho_{12} \otimes \sigma_{34}, \Id_1 \otimes e_{23} \otimes \Id_4 \rangle.
\end{equation}
Similarly, we can instead contract two bipartite effects, partially, with a bipartite state, as depicted in Fig.~\ref{fig:partial_contractions}~(d).
This situation we call \emph{dual entanglement swapping}, and as before, for bipartite effects $e, f$ and a bipartite state $\rho$, we denote it by 
\begin{equation}
    \dswap{e}{\rho}{f} := \langle \Id_1 \otimes \rho_{23} \otimes \Id_4, e_{12} \otimes f_{23} \rangle
\end{equation}

In the following, we generally drop subscripts labelling systems in contractions, if the system an object belongs to is clear from context.

\begin{figure}
    \centering
    (a)
    \includegraphics[width=0.2\textwidth, height=20mm, keepaspectratio]{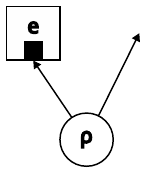}
    \hspace{20mm}
    (b)
    \includegraphics[width=0.2\textwidth, height=20mm, keepaspectratio]{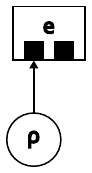}
    \\
    \vspace{5mm}
    (c)
    \includegraphics[width=0.2\textwidth, height=20mm, keepaspectratio]{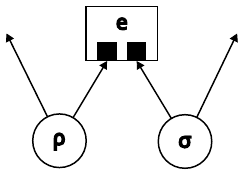}
    \hspace{5mm}
    (d)
    \includegraphics[width=0.2\textwidth, height=20mm, keepaspectratio]{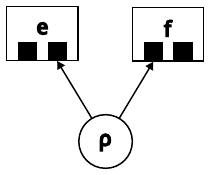}
    \caption{Operational depiction of various partial contractions on the level of a bipartite theory. (a) depicts partial contraction of a unipartite effect and a bipartite state, (b) depicts partial contraction of a bipartite effect and a unipartite state, (c) depicts entanglement swapping, and, (d) depicts dual entanglement swapping}
    \label{fig:partial_contractions}
\end{figure}


\subsection{Outline}

This paper is structured as follows. 
In Section~\ref{sec:theory_axioms} we state our axioms for a theory under which entanglement swapping is well-defined.
We also present an algorithm to either show that a given bipartite theory is inconsistent or extend it to a multipartite theory in a consistent manner.

In Section~\ref{sec:the_question} we discuss the iterated CHSH game.

In Section~\ref{sec:eggs_you_can_yolk} we discuss the notion of composite GPTs, which can swap PR-box correlations for a finite number of steps. 

Finally, in Section~\ref{sec:obl_stab_theory} we present a GPT which sustains a maximal violation of Tsirelson's bound indefinitely under entanglement swapping, and provide an optimal strategy for the iterated CHSH game.


\section{Definition of a Theory}\label{sec:theory_axioms}
The main goal of this paper is to analyse the phenomenon of entanglement swapping in the general context of GPTs.
To this end, it is necessary to specify the conditions on a theory under which entanglement swapping is well defined.

There are a few choices we make therein.
First, we will choose the fundamental objects of our theory to be convex cones rather than starting with the state/effect spaces. 
This is because it is sufficient (and more convenient) to work with cones in the present context, i.e.,
describing notions such as entanglement swapping and partial contractions. 
The state and effects spaces of the theory are obtained as derived objects, in the usual way, as specified below.  

Second, we find it a more natural construction to reinterpret the effects as the primal objects and the states as linear functionals on the effects. 
The two formulations are clearly equivalent. 
For a detailed exposition of such a construction refer to the review \cite{plavala2023general}.

\begin{definition}\label{def:theory}
    A theory in which entanglement swapping is well defined is specified by the following:
    \begin{enumerate}
    \item 
		A finite-dimensional vector space $V$.

	\item 
		An element $\1 \in V$, called the \emph{unit effect}.

	\item A collection of convex cones $\{ \P^{(n)} \}_{n \in \NN}, \ \P^{(n)} \subset V^{\otimes n}$, representing (unnormalised) effects, subject to:
        \begin{enumerate}
            \item $\1 \in \P^{(1)}$.
            \item Closure under tensor products, i.e.,

            If $e \in \P^{(n)}, f \in \P^{(m)}$, then $e \otimes f \in \P^{(n+m)}$.
        \end{enumerate}

	\item A collection of convex cones $\{ \D^{(n)} \}_{n \in \NN}, \ \D^{(n)} \subset (V^{\otimes n})^*$, representing (unnormalised) states, subject to:
        \begin{enumerate}
	    \item Positivity, i.e., $\D^{(n)} \subset (\P^{(n)})'$, the polar dual of $\P^{(n)}$.
            \item Closure under tensor products, i.e., 

            If $\rho \in \D^{(n)}, \sigma \in \D^{(m)}$, then $\rho \otimes \sigma \in \D^{(n+m)}$
        \end{enumerate}

        \item Closure under \emph{partial contractions}, i.e., 

		For $e \in \P^{(n)}, \rho \in \D^{(m)}$:
        \begin{enumerate}
            \item If $n > m$, then $\langle \rho, e \rangle \in \P^{(n-m)}$
            \item If $n < m$, then $\langle \rho, e \rangle \in \D^{(m-n)}$
            \item (for completeness) If $n = m$, then $\langle \rho, e \rangle \in \RR^+$, is the canonical pairing.
        \end{enumerate}

	\item Invariance under permutations of systems, i.e., for every $n \in \NN$ and $\pi \in S_n$, we have $\pi(\D^{(n)}) = \D^{(n)}$ and $\pi(\P^{(n)}) = \P^{(n)}$. Where, for $\rho_{1 \cdots n} \in \D^{(n)}$
    \begin{equation*}
        \pi(\rho_{1 \cdots n}) := \rho_{\pi(1) \cdots \pi(n)}. 
    \end{equation*}
    Similarly for effects.
    \end{enumerate}
\end{definition}

From the above definition, we can derive the following: 

For an effect $e \in \P^{(n)}$, define the \emph{negation}
\begin{equation}
    \neg e := \1^{\otimes n} - e.
\end{equation}

Define the \emph{$n$-partite effect space} as
\begin{equation}
    \E^{(n)} := \P^{(n)} \cap \neg(\P^{(n)}).
\end{equation}

Define the \emph{$n$-partite state space} as 
\begin{equation}
    \S^{(n)} := 
    \{
    \rho \in \D^{(n)}    
    \ | \
    \langle \rho, \1^{\otimes n} \rangle 
    =
    1
    \}.
\end{equation}

Define the set of \emph{unipartite correlators} as
\begin{equation}
    \mathcal{X}^{(1)}
    :=
    \{
    e - \neg e
    \ | \  
    e \in \E^{(1)}
    \}.
\end{equation}

For every choice $A_0, A_1, B_0, B_1 \in \mathcal{X}^{(1)}$, we define the \emph{CHSH observable} as
\begin{equation}
    \begin{aligned}\label{eqn:chsh_obs}
        &\mathrm{CHSH}(A_0, A_1 ; B_0, B_1) 
        \\
        &:= 
        A_0 \otimes B_0 + A_0 \otimes B_1 + A_1 \otimes B_0 - A_1 \otimes B_1.
    \end{aligned}
\end{equation} 
The \emph{CHSH value} corresponding to the above choice of correlators and some choice of $\rho \in \S^{(2)}$ is the expectation value of the CHSH observable with respect to $\rho$, i.e.,
\begin{equation}\label{eqn:chsh_val}
    \langle 
        \rho, 
        \mathrm{CHSH}(A_0, A_1 ; B_0, B_1)
    \rangle.
\end{equation}
We can associate with every theory a CHSH value defined to be the supremum over all possible choices $\rho, A_0, A_1, B_0, B_1$. 
\begin{equation}\label{eqn:theory_chsh}
    \underset{
    \substack{
        A_0, A_1, B_0, B_1 \in \mathcal{X}^{(1)} 
        \\ 
        \rho \in \S^{(2)}
        }
    }
    {\mathrm{sup}}
    |\langle 
        \rho, \mathrm{CHSH}(A_0, A_1 ; B_0, B_1)
    \rangle|.
\end{equation} 

Another important notion in the discussion that follows, is that of \emph{closure under entanglement swapping}.
Given a theory with bipartite cones $\P^{(2)}$ and $\D^{(2)}$, we say $\D^{(2)}$ is \emph{closed} under entanglement swapping if 
\begin{equation}
    \swap{\D^{(2)}}{\P^{(2)}}{\D^{(2)}} \subset \D^{(2)}.
\end{equation} 
Further, we say $\D^{(2)}$ is \emph{stable} under entanglement swapping if 
\begin{equation}
    \mathrm{conv}\Big(\swap{\D^{(2)}}{\P^{(2)}}{\D^{(2)}}\Big) = \D^{(2)}.    
\end{equation}
The same notions apply to dual entanglement swapping.


\subsection{Consistent Bipartite Theory}\label{subsec:consistent_bip_theory}

In this section, we are concerned with the following question. 
Given:
\begin{itemize}
	\item
		A finite-dimensional vector space $V$,
	\item
		a non-zero element $\1\in V$, 
	\item
		a convex cone $P \subset V\otimes V$, 
	\item
		a convex cone $D \subset (V\otimes V)^*$, 
	\end{itemize}
	is it possible to produce a multipartite theory which has the same unit effect, for which $\P^{(2)} = P$ and $\D^{(2)} = D$?

We answer this question by constructing an algorithm that, for an input $n \in \NN$, either produces $n$-partite cones $\P^{(n)}$ and $\D^{(n)}$ by extending $P$ and $D$, if consistently possible, or detects the inconsistency.
The algorithm can be split into two parts. First, a consistency check on the objects $\1, P$ and $D$, followed by an explicit construction of $\P^{(n)}$ and $\D^{(n)}$ for every $n \in \NN$. The algorithm is as follows:

\begin{algorithm}[H]
    \caption{Induced theory}\label{alg:induced_theory}
    \begin{algorithmic}[1]
        \State \textbf{input:} $(V, \1, P, D, n \in \NN)$
        \If{\Call{check consistency}{$V,\1, P, D$} = ``Inconsistent''}
            \State\Return ``Inconsistent''
        \Else 
            \State $\D^{(1)}\leftarrow \mathrm{cone}\langle D, \1 \rangle$
            \Comment{\# partial contractions}
            \State $\P^{(1)}\leftarrow \mathrm{cone}\langle\D^{(1)}, P\rangle$
            \State $\P^{(2)}\leftarrow P$
            \State $\D^{(2)}\leftarrow D$
            \If{$n$ is even}
                \State 
                $\P^{(n)}
                \leftarrow
                S_{n} \,\cdot\,
                \Big(
                    {\P^{(2)}}^{\dot\otimes n/2}
                \Big)$ 
                \State 
                $\D^{(n)}
                \leftarrow
                S_{n} \,\cdot\,
                \Big(
                    {\D^{(2)}}^{\dot\otimes n/2}
                \Big)$
            \Else
                \State 
                $\P^{(n)}
                \leftarrow
                S_{n} \,\cdot\,
                \Big(
                \P^{(1)} \dot\otimes
                {\P^{(2)}}^{\dot\otimes \lfloor n/2 \rfloor}
                \Big)$ 
                \State 
                $\D^{(n)}
                \leftarrow
                S_{n} \,\cdot\,
                \Big(
                \D^{(1)} \dot\otimes
                {\D^{(2)}}^{\dot\otimes \lfloor n/2 \rfloor}
                \Big)$ 
            \EndIf
            \State\Return $\P^{(n)}, \D^{(n)}$
        \EndIf
    \end{algorithmic}
    \vspace{10pt}
    \begin{algorithmic}[1]
        \algblockdefx{Conditions}{EndConditions}{\textbf{if}}{\textbf{then}}
        \Function{check consistency}{$V,\1, P, D$}
            \State $\D^{(1)}\leftarrow \mathrm{cone}\langle D, \1 \rangle$
            \Comment{\# partial contractions}
            \State $\P^{(1)}\leftarrow \mathrm{cone}\langle\D^{(1)}, P\rangle$
            \Conditions
                \State $\1\not \in \P^{(1)}$  
                \Comment{\# Axiom~3(a)}

                \State \textbf{or}

                \State $\langle \D^{(1)}, \P^{(1)}\rangle < 0$ 
                \Comment{\# positivity, unipartite}

                \State \textbf{or}

                \State $\langle D, P\rangle < 0$
                \Comment{\# positivity, bipartite}

                \State \textbf{or}

                \State $\forall \ \pi \in S_2, \ \pi(P) \neq P$
                \Comment{\# Axiom~6, bipartite}

                \State \textbf{or}

                \State $\forall \ \pi \in S_2, \ \pi(D) \neq D$

                \State \textbf{or}

                \State $\P^{(1)} \otimes \P^{(1)}  \not\subset P$
                \Comment{\# Axiom~3(b), unipartite}

                \State \textbf{or}

                \State $\D^{(1)} \otimes \D^{(1)}  \not\subset D$
                \Comment{\# Axiom~4(b), unipartite}

                \State \textbf{or}

                \State $\swap{D}{P}{D}\not\subset D$
                \Comment{\# Axiom~5(b), bipartite}

                \State \textbf{or}

                \State $\dswap{P}{D}{P}\not\subset P$
                \Comment{\# Axiom~5(a), bipartite}
            \EndConditions
            \State \qquad \Return ``inconsistent''
            \State \textbf{else}
            \State \qquad \Return ``consistent''
        \EndFunction
    \end{algorithmic}
\end{algorithm}

In the presentation of Algorithm~\ref{alg:induced_theory}, $\dot\otimes$ stands for the \emph{minimal tensor product}, defined as the conal hull over the tensor product of the respective sets:
\begin{equation}
    \mathcal P^{(n)} \dot\otimes \mathcal P^{(m)}
    :=
    \mathrm{cone}\big(\mathcal P^{(n)} \otimes \mathcal P^{(m)}\big).
\end{equation}
Additionally, we have not specified in which format inputs like $V$ or $P$ are to be supplied, or how the checks should be performed.
In this sense, it is a ``template'' for a concrete algorithm that depends on the mathematical properties of the input data.
For example, if $P, D$ are polyhedral cones in $\RR^d$, then $V$ can be represented by the integer $d$ and the cones by the facet inequalities.
In this case, all tests in Algorithm~\ref{alg:induced_theory} are linear programs.
But more general situations also make sense, e.g.\ cones defined in terms of semi-definite constraints.

\begin{lemma}\label{lem:no_abrt_impls_theory}
    Given the input data $V, \1, P, D$, Algorithm~\ref{alg:induced_theory} either detects inconsistency, or returns consistent $n$-partite cones $\P^{(n)}$ and $\D^{(n)}$ induced by the data, for any $n \in \NN$, in a finite number of steps.
\end{lemma}

\begin{proof}
	If the function \emph{check consistency} in Algorithm~\ref{alg:induced_theory} returns ``inconsistent'',  then one of the axioms in Def.~\ref{def:theory} is violated already at the level of bipartite objects and so there is no consistent extension of the input.

	We will now verify that if the function \emph{check consistency} in
	Algorithm~\ref{alg:induced_theory} returns ``consistent'',  
	then there exists a theory as advertised.

	Indeed, the objects whose existence is posited by Axioms~1 and~2 are part of the problem data. 
    Axiom~6 holds by construction given that the bipartite cones are invariant under permutations. 
    Axiom~3(a) is checked directly.
    The $n$-partite cones are constructed only using the unipartite and bipartite cones.
    Therefore, Axioms~5(a) and~5(b) follow from Axiom~6 and closure 
    under partial contractions, entanglement swapping and dual entanglement swapping, at the bipartite level. 
    Axiom~4(a) follows from Axioms~5(a) and~5(b) plus the positivity of the unipartite and bipartite cones, since tensor products preserve positivity. 
    Axioms~3(b) and~4(b) follow by construction.  
\end{proof}


\section{The Iterated CHSH Game}\label{sec:the_question}

The iterated CHSH game, as previously alluded to, is an extension of the adaptive CHSH game from Ref.~\cite{weilenmann2020self}, to include multiple rounds of entanglement swapping.
The game can be viewed as implementing a CHSH test after the use of a ``GPT repeater'', the generalisation of a quantum repeater to GPTs, in order to probe the capacity of the repeater to propagate entanglement.

The iterated CHSH game is parameterised by $n$, referring to the number of repeater units between the start and end nodes.
The players of the game are Alice ($A$), a collection of $n$ Bobs ($\{ B_i \}_{i = 1, \cdots, n}$), and Charlie ($C$).
Alice corresponds to the start node, Charlie to the end node, and the Bobs to repeaters.
The diagramatic reperesentation of the structure thus produced is shown in Fig.~\ref{fig:iterated_CHSH}

\begin{figure}
    \centering
    \includegraphics[width=0.45\textwidth]{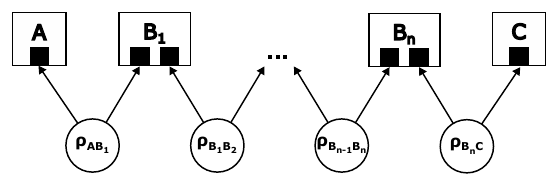}
    \caption{
    Diagrammatic representation of the structure of the iterated CHSH game. 
    $A$ and $C$ stand for Alice and Charlie respectively. 
    The first and the last Bob are shown as $B_1$ and $B_n$ respectively. 
    The other Bobs are similarly iterated in order. 
    Each of the nearest neighbours shares a bipartite resource. 
    The local effect space of all the Bobs is bipartite, and that of Alice and Charlie is unipartite.}
    \label{fig:iterated_CHSH}
\end{figure}

Each round of the game proceeds as follows.
First, each of the Bobs performs a bipartite measurement on the subsystems available to him, effectively leading to $n$ rounds of entanglement swapping. 
They subsequently broadcast their outcomes.
Alice and Charlie are then allowed to implement local corrections based on the outcomes of the Bobs.
Finally, Alice and Charlie perform a CHSH test. 
In addition to shared randomness among all parties, each of the nearest neighbours depicted in Fig.~\ref{fig:iterated_CHSH} may share a bipartite resource, but no other multipartite resource. 
(The adaptive CHSH game of Ref.~\cite{weilenmann2020self}, depicted in Fig.~\ref{fig:adaptive_CHSH}, is the iterated CHSH game for $n = 1$).

\begin{figure}
    \centering
    \includegraphics[width=0.25\textwidth]{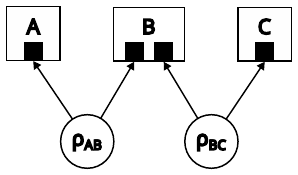}
    \caption{Diagrammatic representation of the structure of the adaptive CHSH game. $A, B$ and $C$ stand for Alice, Bob and Charlie respectively. Bob shares a bipartite resource $\rho_{AB}$ with Alice and $\rho_{BC}$ with Charlie.}
    \label{fig:adaptive_CHSH}
\end{figure}

To calculate the resulting CHSH value between Alice and Charlie in the general case, we simply sum over the CHSH value corresponding to each outcome vector $\vec b = (b_1, \cdots, b_{n})\in (\ZZ_k)^{n}$ (for a $k$-outcome measurement) produced by the Bobs, weighted by the respective probabilities.
That is, if $\beta_{\vec b}$ and $p_{\vec b}$ are the CHSH value and probability corresponding to the outcome vector $\vec b$ respectively, the CHSH value of the game is calculated as 
\begin{equation}\label{eqn:I_CHSH_val}
    \beta = \sum_{\vec b \in (\ZZ_k)^{n}} p_{\vec b} \beta_{\vec b}.
\end{equation}

In the case of quantum mechanics, it is trivial to describe an optimal strategy for the iterated CHSH game:
Choose all the shared states to be the Bell state $|\Phi^+\rangle$.
The Bobs perform a Bell basis measurement each.
Conditioned on the outcomes, the state shared by Alice and Charlie is then one of four elements of the Bell basis, all of which are equiprobable.
A local Pauli operation, performed by either of them, can map it to $|\Phi^+\rangle$.
Hence, it follows that quantum theory retains a CHSH value of $2\sqrt2$ for all $n$.


\section{Composite GPTs}\label{sec:eggs_you_can_yolk}

The present section is inspired by private communications with 
Roger Colbeck,
Marc-Olivier Renou,
Mirjam Weilenmann,
and
Elie Wolfe. 

Here we discuss a family of constructions that can each swap PR-box correlations for a fixed number of iterations $m$.
That is for all $n \leq m$, they have a CHSH value of $4$ in the iterated CHSH game.
But, as we will also show, for all $n > m$ they have a CHSH value of $2$.

The main idea of this construction is to use particles with multiple ``internal degrees of freedom''.
Roughly speaking, we assign to some d.o.f.'s the task of carrying the entanglement, and to others the task of supporting entangled measurements.
As described below, this way, we can circumvent the tension between state and effect space sizes.
We call these particles \emph{composites}, as they can be viewed as a composite of particles from smaller GPTs.

\subsection{An example}

Consider the following example.
Say we have particles with two d.o.f.'s, labeled $1$ and $2$.
Each d.o.f.\ supports the same local effects and local states as a unipartite boxworld theory (\cite{plavala2023general}).
That is, if $\P$ and $\D$ are the unipartite boxworld state and effect cones, then the effects measurable on each composite particle are
\begin{equation}
	\P_{c} 
	=
	\P_{1} \dot\otimes \P_2,
\end{equation}

Physically, this says that we do not allow for entangled measurements between the d.o.f.'s within each composite particle.
We define the joint state space within each particle in the same way:
\begin{equation}
	\D_{c} 
	=
	\D_1 \dot\otimes \D_2,
\end{equation}

We now define the effects and states realisable between two composite particles.
The $n$-partite theory then results from this input data via the general construction of Sec.~\ref{subsec:consistent_bip_theory}.

Let
\begin{equation}
	\widehat{\D}
	=
	\D
	\hat\otimes
	\D
\end{equation}
be the cone of bipartite boxworld states.
Here, $\hat\otimes$ is the \emph{maximal tensor product}.
Mathematically, it is the \emph{polar dual} to the minimal tensor effects:
\begin{equation}
    \mathcal D \hat\otimes \mathcal D := (\mathcal P \dot\otimes \mathcal P)'
\end{equation}
Physically it gives rise to the maximally entangled boxworld states that achieve CHSH values of $4$ \cite{popescu1994quantum, barrett2007information, plavala2023general}.

In our theory, in addition to product states, we allow for maximally entangled boxworld states between the first d.o.f.'s of any two composite particles, and between the first d.o.f.\ of one and the second d.o.f.\ of the other.
We thus arrive at the following cone of composite bipartite states
\begin{equation}
    D_c
    :=
    \mathrm{cone}\big(
        \widehat{\D}_{1_A,1_B}
        \otimes
        \D_{2_A}
        \otimes
        \D_{2_B}
        \cup
        \widehat{\D}_{1_A,2_B}
        \otimes
        \widehat{\D}_{2_A,1_B}
    \big),
\end{equation}
where we have indicated the tensor factors that any object acts on in the superscripts:
Letters $A$, $B$ refer to the first and second composite particle, respectively;
and numbers $1$, $2$ to the internal degrees of freedom.

Finally, the bipartite effects are just the minimal tensor product of the local ones,
except that we allow for maximally entangled measurements between the second d.o.f.'s of any two composite particles (This is consistent because no entanglement exists between these particular d.o.f.'s).
As in the case of states, denote the cone of maximally entangled bipartite effects as
\begin{equation}
	\widehat\P
	=
	\P
	\hat\otimes
	\P.
\end{equation}

Then, for the composite bipartite effects we get:
\begin{equation}
    P_c
    :=
    \mathrm{cone}\big(
        \P_{1_A}
        \otimes
        \P_{1_B}
        \otimes
        \widehat{\P}_{2_A,2_B}
    \big).
\end{equation}

One can now check that $P_c, D_c$ consistently define a theory.
This theory allows for a CHSH value of $4$ after one round of entanglement swapping, according to the scheme visualized in Fig.~\ref{fig:composite_ent_swapping_caric}.

\begin{figure}
    \centering
    \includegraphics[width = 0.45\textwidth]{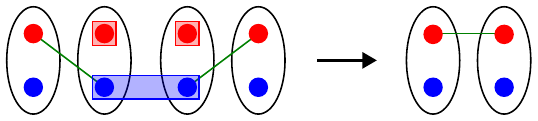}
    \caption{Caricature of PR-box entanglement swapping using composite particles.
    The red dot signifies d.o.f.~$1$, the blue dot d.o.f.~$2$.
    The green lines signify entangled states. 
    The blue rectangle signifies an entangled measurement, the red squares signify product measurements.}
    \label{fig:composite_ent_swapping_caric}
\end{figure}

\subsection{General construction}
	
In the example just treated, by adding an additional d.o.f., we managed to move a PR box past one round of entanglement swapping.
It is resonable to conjecture that with access to more d.o.f.'s one may be able to find a way to sustain entanglement longer, maybe even indefinitely.

To investigate this, we generalise the previous example to $m$ degrees of freedom. 
Each d.o.f. supports unipartite boxworld states and effects.
Now, for each d.o.f. $k$ we have to specify the two disjoint subsets of d.o.f.'s that form entangled states and entangled effects with $k$.
This amounts to specifying two symmetric bipartite graphs as follows:

Start with a collection of $2m$ vertices, two per d.o.f.
\begin{equation}
    \begin{aligned}
        V 
        &=
        \{ v_1, \cdots, v_m, w_1, \cdots, w_m \} 
        \\
        & = 
        \{ v_i \}_{i = 1, \cdots, m} \cup \{ w_i \}_{i = 1, \cdots, m}.
    \end{aligned}
\end{equation}

On this vertex set, define two symmetric bipartite graphs, $G = (V, E)$ and $H = (V, F)$ (with the same bipartition over the above indicated subsets), such that they lie in each other's complement.
The elements of edge sets $E$ and $F$ therefore, are ordered pairs of vertices, one from each of the subsets $\{v_i\}_{i = 1, \cdots, m}$ and $\{ w_i\}_{i = 1, \cdots, m}$.

Indeed, the edges of graph $G$ specifies those pairs of d.o.f.'s that can support entangled states, whereas the edges of graph $H$ give those pairs that can support entangled measurements.

The entire construction can be summarised by the following rules imposed on the edge sets $E$ and $F$:
\begin{enumerate}
    \item If $(v_i, w_j) \in E$ then $(v_j, w_i) \in E$, by symmetry.
    Same for $F$.
    \item If $(v_i, w_j) \in E$ then $(v_i, w_j) \not\in F$, because they lie in each other complements
    \item If $(v_i, w_j), (v_k, w_l) \in E$ and $(v_j, w_k) \in F$ then, we can concatenate edges via entanglement swapping as follows:
    \begin{equation*}
        (v_i, w_j)\circ(v_j, w_k)\circ(v_k, w_l) \equiv (v_i, w_l) \in E.
    \end{equation*}
    \item If $(v_i, w_j), (v_k, w_l) \in F$ and $(v_j, w_k) \in E$ then, we can concatenate edges via dual entanglement swapping as follows:
    \begin{equation*}
        (v_i, w_j)\circ(v_j, w_k)\circ(v_k, w_l) \equiv (v_i, w_l) \in F.
    \end{equation*}
\end{enumerate}

\begin{lemma}
    Given a composite particle with $m$ boxworld d.o.f.'s, there exists $l \in \NN$ such that $\forall \, n \geq l$ the CHSH value of the composite GPT in the iterated CHSH game parametrised by $n$, is $2$.
\end{lemma}

\begin{proof}
    In order to sustain entanglement indefinitely under entanglement swapping with a finite number of d.o.f.'s, we must have a ``closed cycle''.
    That is, there is some chain of concatenations, alternating edges from $E$ and $F$ (starting and ending with $E$) that reproduces the first edge.
    In equations:
    \begin{equation}
        \begin{aligned}
            &(v_{1}, w_{2}) \circ (v_{2}, w_{3}) \circ \cdots
            \\
            & \quad
            \cdots \circ (v_{k-1}, w_{k}) \circ (v_{k}, w_{2}) \equiv (v_{1}, w_{2}) \in E,
        \end{aligned}
    \end{equation}
    with $(v_{i}, w_{i+1})$ belongs to $E$ for odd $i$ and $F$ for even $i$.  

    Since the second and penultimate edge belong to $F$, we can use Rule.~4 to simplify the above chain to
    \begin{equation}
        (v_{1}, w_{2}) \circ (v_{2}, w_{k}) \circ (v_{k}, w_{2}) \equiv (v_{1}, w_{2})
    \end{equation}

    The above equation implies that $(v_{1}, w_{2}), (v_{k}, w_{2}) \in E$ and $(v_{2}, w_{k}) \in F$.
    But, by Rule.~1 we get $(v_{k}, w_{2}) \in F$ which contradicts Rule.~2.

    This then means, each time we do an additional round of entanglement swapping, we have to add a new degree of freedom (if not more).

    This implies, given access to $m$ d.o.f.'s, the maximum number of entanglement swapping rounds such that the output is still entangled is less than or equal to $m-1$.

    Therefore, for any $n \geq m$ we end up with only product states in the iterated CHSH game, meaning the CHSH value is no more than $2$.
\end{proof}


\section{Oblate Stabilizer Theory}\label{sec:obl_stab_theory}

In this section we present our main result.
That is, we construct the \emph{oblate stabilizer theory}, which  not only achieves a CHSH value of 4 but is also stable under entanglement swapping. 
In other words, it can sustain this CHSH value indefinitely under entanglement swapping.
We also show that, given the resources described by this theory, there is an optimal strategy by which we get a CHSH value of 4 in the iterated CHSH game.

\subsection{Setup}

The theory can be obtained by slightly deforming the set of quantum-mechanical stabilizer states.
For this reason, we will use objects from the mathematical description of quantum mechanics to construct it.

Consider the one-qubit stabilizer polytope, i.e., the convex hull of the following states on the Bloch sphere 
\begin{equation}
    \begin{gathered}
        | x_\pm \rangle\langle x_\pm | := \tfrac12(\Id \pm \sigma_1), 
        \\
        | y_\pm \rangle\langle y_\pm | := \tfrac12(\Id \pm \sigma_2),
        \\
        | z_\pm \rangle\langle z_\pm | := \tfrac12(\Id \pm \sigma_3),
    \end{gathered}
\end{equation}
where $\sigma_1, \sigma_2, \sigma_3$ are the Pauli matrices.
Now perform a ``uniform stretch in the equatorial plane of the Bloch sphere'', i.e., for some $r > 1$, set 
\begin{equation}
    \begin{gathered}
        | \tilde x_\pm \rangle\langle \tilde x_\pm | := \tfrac12(\Id \pm r\sigma_1), 
        \\
        | \tilde y_\pm \rangle\langle \tilde y_\pm | := \tfrac12(\Id \pm r\sigma_2),
        \\
        | \tilde z_\pm \rangle\langle \tilde z_\pm |
        :=
        | z_\pm \rangle\langle z_\pm | = \tfrac12(\Id \pm \sigma_3).
    \end{gathered}
\end{equation}
These will be the building blocks for the unipartite state space.
The \emph{Bell state} is defined as usual as:
\begin{equation}
    | \Phi^+ \rangle\langle \Phi^+ |
    :=
    \tfrac14 
    (\sigma_0 \otimes \sigma_0
    +\sigma_1 \otimes \sigma_1
    -\sigma_2 \otimes \sigma_2
    +\sigma_3 \otimes \sigma_3).
\end{equation}
It satisfies the standard identities 
\begin{equation}\label{eqn:id_Bell_shift}
    (A \otimes B)| \Phi^+ \rangle = (\Id \otimes BA^t)|\Phi^+\rangle
\end{equation}
and
\begin{equation}\label{eqn:id_Bell_part_tr}
    \tr_1(| \Phi^+ \rangle\langle \Phi^+ |) 
    = \tr_2(| \Phi^+ \rangle\langle \Phi^+ |)
    = \tfrac12 \Id.
\end{equation}
Where $A^t$ stands for the matrix-transpose of $A$, and 
\begin{equation}
    | \Phi^+ \rangle = \tfrac1{\sqrt2}(|00\rangle + |11\rangle).
\end{equation}
From the Bell state, we can generate the \emph{Bell basis} as:
\begin{equation}\label{eqn:Bell_basis}
    | \Phi^+_\mu \rangle\langle \Phi^+_\mu |
    :=
    (\sigma_\mu \otimes \Id)| \Phi^+ \rangle\langle \Phi^+ |(\sigma_\mu \otimes \Id),
    \quad
    \mu \in \ZZ_4.
\end{equation}

In keeping with the quantum formalism, effects will also be represented by $2 \times 2$ matrices
and the pairing between states and effects by the trace inner product. 
Additionally, partial contractions between states and effects are realised as partial traces. 
For example, for a unipartite effect $e_1$ and bipartite state $\rho_{12}$, the partial contraction is
\begin{equation}
    \langle \rho_{12}, e_1 \rangle := \tr_1\big(\rho_{12}(e_{1} \otimes \Id)\big).
\end{equation}
Similarly, for a bipartite effect $e$ and bipartite states $\rho$ and $\sigma$, the entanglement swapping map is
\begin{equation}
    \swap{\rho}{e}{\sigma}
    :=
    \tr_{23}
    \big( 
    \rho_{12} \otimes \Id^{\otimes 2} 
    (\Id \otimes e_{23} \otimes \Id) 
    \Id^{\otimes 2} \otimes \sigma_{34} 
    \big).
\end{equation} 

\subsection{Unipartite Theory}\label{subsec:unip_obl_stab_theory}

Let $R = e^{-i \frac\pi8 \sigma_3}$ be the unitary which, by conjugation ($R(\,\cdot\,)R^\dag$), implements a $\pi/4$-rotation of the Bloch sphere about the $z$-axis.

Let $\Omega$ be the stretched stabilizer states, for some choice of $r > 1$ :
\begin{equation}
    \Omega := 
    \{ 
        | \tilde x_\pm \rangle\langle \tilde x_\pm |,
        | \tilde y_\pm \rangle\langle \tilde y_\pm |,
        | \tilde z_\pm \rangle\langle \tilde z_\pm |
    \}
\end{equation}

\begin{definition}[Unipartite Oblate Stabilizer Theory]\label{def:loc_obl_stab_th}
    Choose $r = \big(\cos(\tfrac\pi4)\big)^{-\frac12} = \sqrt[4]2$, and define the following:
    \begin{enumerate}
        \item The unit effect, $\1 := \Id$, the $2 \times 2$ identity matrix.
        \item The effect cone, $\P^{(1)} := \mathrm{cone}\big(R \Omega R^\dag \big)$.
        \item The state cone, $\D^{(1)} := \mathrm{cone}\big( \Omega \big)$.
    \end{enumerate}
\end{definition}

\begin{figure}
    \centering
    \includegraphics[width=0.45\textwidth]{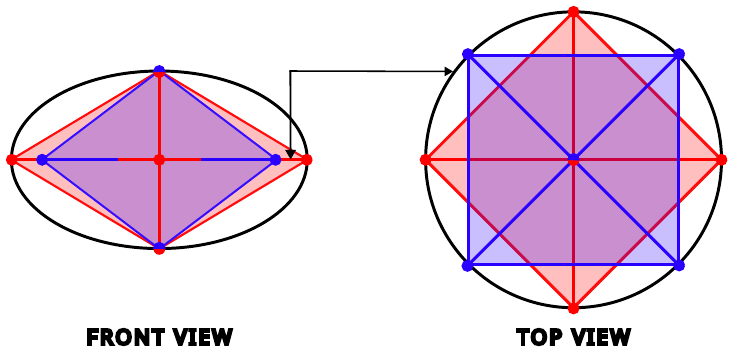}
    \caption{Caricature of the geometry of the unipartite Oblate Stabilizer Theory. The black outlines represent a Bloch sphere with a scaled up equatorial plane. The rays formed by vertices of the red polytope (what used to be the stabilizer states) represent the extremal rays of the state cone. The rays formed by vertices of the rotated (about the $z$-axis, by $\pi/4$) blue polytope represent the extremal rays of the effect cone.}
    \label{fig:obl_stab_th}
\end{figure}

Indeed, we get the following derived objects: The negation of an effect is
\begin{equation}
    \neg e  = \Id - e.
\end{equation}
The state space is simply 
\begin{equation}
    \S^{(1)} = \mathrm{conv}(\Omega).
\end{equation}
Similarly, the effect space is 
\begin{equation}
    \E^{(1)} = \mathrm{conv}(\{ 0, \Id \} \cup R \Omega R^\dag).
\end{equation}

It is easily verified that we have a well-defined unipartite theory as:

\noindent 1. The effect space is closed under negations
\begin{equation}\label{eqn:closed_und_negs}
    \neg \E^{(1)} = \E^{(1)}.
\end{equation}
2. The states are normalised
\begin{equation}
    \rho \in \S^{(1)} \quad \Rightarrow \quad \tr(\rho \Id) = 1.
\end{equation}
3. Pairing a state and an effect gives a probability, i.e.,
\begin{equation}\label{eqn:uni_pairing_givs_probs}
    \rho \in \S^{(1)}, \ e \in \E^{(1)} \quad \Rightarrow \quad \tr(\rho e) \in [0,1].
\end{equation}
It suffices to check~(\ref{eqn:uni_pairing_givs_probs}) for the extremal vertices. 
Additionally, due to~(\ref{eqn:closed_und_negs}) only the upper bound needs to be checked. 
From Fig.~\ref{fig:obl_stab_th}, the inner product on the $z$-axis is unchanged from quantum theory and hence bounded between $0$ and $1$. 
On the equatorial plane the largest inner product is between two nearest neighbours. 
By construction, the angle enclosed by the corresponding Bloch vectors is $\pi/4$, meaning
\begin{equation}
    \underset{\rho, e}{\mathrm{sup}} \tr(\rho e) = \tfrac12\big( 1 + r^2 \cos(\pi/4) \big) = 1.
\end{equation}

\subsection{Bipartite Theory}\label{subsec:bip_OST}

To construct the bipartite theory we first interpret the bipartite states as maps from the unipartite effect cone to the unipartite state cone, via partial contraction. 
Similarly for bipartite effects (refer Fig.~\ref{fig:partial_contractions}~(a)~(b)). Then, we use the following properties of the unipartite theory:
\begin{enumerate}
    \item Both the state and effect cones are invariant under conjugation by Pauli matrices. 
    In the Bloch picture, these correspond to reflections about the $x$, $y$, and $z$-axes.
    \item Both state and effect cones are invariant under matrix transpose. 
    In the Bloch picture, this corresponds to  reflections about the $xz$-plane.
    \item Both the state and effect cones are invariant under conjugation by $R^m$ for $m \in \{ 0,2,4,6 \}$.
    In the Bloch picture this corresponds to a $m\pi/4$-rotation about the $z$-axis.
    \item On the other hand, conjugation by $R^m$ for $m \in \{1,3,5,7\}$ maps the state cone to the effect cone and vice-versa.
\end{enumerate}

\noindent Therefore, for $\mu \in \ZZ_4$, $m \in \ZZ_8$ define the projections
\begin{equation}
    | \Phi^+_{\mu,m} \rangle \langle \Phi^+_{\mu,m} |
    :=
    (\sigma_\mu R^m)^\dag \otimes \Id 
    | \Phi^+ \rangle \langle \Phi^+ |
    (\sigma_\mu R^m) \otimes \Id.
\end{equation}
Further, define the set 
\begin{equation}
    \Phi := 
    \big\{ 
    | \Phi^+_{\mu,m} \rangle \langle \Phi^+_{\mu,m} | 
    \ \big| \ 
    \mu \in \ZZ_4, \ m \in \ZZ_8, \ m \ \text{odd}
    \big\}.
\end{equation}

\begin{definition}[Oblate Stabilizer Theory]\label{def:bip_obl_stab_th}
    In the sense of Sec.~\ref{subsec:consistent_bip_theory}, \emph{oblate stabilizer theory} is the theory specified by the following data:
    \begin{enumerate}
        \item $V$, the set of $2\times 2$ Hermitian matrices,
        \item $\1=\Id$, the $2 \times 2$ identity matrix,
        \item $P = \mathrm{cone}\big( R\Omega R^\dag \otimes R\Omega R^\dag \cup \Phi \big)$,
        \item $D = \mathrm{cone}\big( \Omega \otimes \Omega \cup \Phi \big)$.
    \end{enumerate}
\end{definition}

\begin{lemma}
    The data specified in Def.~\ref{def:bip_obl_stab_th}, can be consistently extended to a theory using Alg.~\ref{alg:induced_theory}.
\end{lemma}

\begin{proof}
	We verify that Algorithm~\ref{alg:induced_theory} accepts the data, which is thus consistent by Lemma~\ref{lem:no_abrt_impls_theory}.

    The partial trace of all entangled states introduced is $\tfrac12 \Id$.
    Partial contractions of entangled effects with unipartite states can be viewed as conjugation by a odd rotation, and then by a Pauli matrix. 
    This, as discussed before, maps the unipartite state cone to the unipartite effect cone.
    Therefore, by construction, the $\D^{(1)}$ and $\P^{(1)}$ assigned by the algorithm are the same $\D^{(1)}$ and $\P^{(1)}$ as in Def.~\ref{def:loc_obl_stab_th}.
    Hence, we already have $\1 \in \P^{(1)}$ and $\langle \D^{(1)}, \P^{(1)}\rangle \geq 0$.

    Since the entangled states and effects are projectors taken from quantum theory without modification, the trace inner product between them is non-negative.

	We have verified above that pairing unipartite states and effects leads to positive outcomes, and this property is preserved under tensor products.

    Using identities~(\ref{eqn:id_Bell_shift})~and~(\ref{eqn:id_Bell_part_tr}), it can be easily shown that 
    \begin{equation}
        \tr\big(| \Phi^+_{\mu,m} \rangle \langle \Phi^+_{\mu,m} | e \otimes f \big)
        =
        \tfrac12 \tr\big(f\sigma_\mu R^m e^t (R^{m})^\dag \sigma_\mu \big),
    \end{equation}
    which is just $\tfrac12$ times the pairing between a unipartite state and effect, and hence, is positive. 
    The same argument extends to entangled effects and product states.
    Therefore, we have $\langle D,P \rangle \geq 0$.

    For invariance under permutation of systems, we need only check $S_2$ invariance for the set of entangled states and effects, since everything else is permutation invariant by construction.
    This is readily verified by using identity~(\ref{eqn:id_Bell_shift}) since $|\Phi^+ \rangle \langle \Phi^+ |$ is already $S_2$ invariant.

    The tensor products of the unipartite cones, $\P^{(1)} \otimes \P^{(1)}$  and $\D^{(1)} \otimes \D^{(1)}$ 
	are subsets of, respectively, $P$ and $D$ by construction.

	It remains to be shown that the theory is closed under entanglement swapping and dual entanglement swapping.
	We only treat the first case explicitly. 
	The dual version follows in complete analogy.
    We separate the entanglement swapping contractions into four types.
	\begin{enumerate}
		\item
		The effect factorizes, i.e., $\swap{\bullet}{e \otimes f}{\bullet}$.
        In this case the result is an element of $\D^{(1)}\otimes\D^{(1)}$, because the contraction splits as follows:
        \begin{equation*}
            \includegraphics[width=0.325\linewidth]{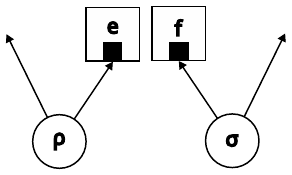}
        \end{equation*}

	    \item
		Both states are product states, i.e., 
        $\swap{\rho_1 \otimes \rho_2}{\bullet}{\rho_3 \otimes \rho_4 }$.
        In this case the result is again an element of $\D^{(1)}\otimes\D^{(1)}$, because the central objects can be grouped as:
        \begin{equation*}
            \includegraphics[width=0.325\linewidth]{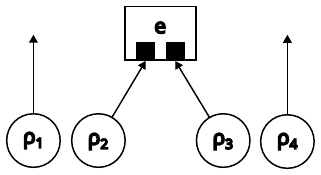}
        \end{equation*}

		\item
		There is one entangled state and one entangled effect, i.e.,
		$\swap{\Phi}{\Phi}{\rho \otimes \sigma}$ for example.
		In this case the result is once again an element of $\D^{(1)}\otimes\D^{(1)}$, because we can split up the contraction as:
        \begin{equation*}
            \includegraphics[width=0.325\linewidth]{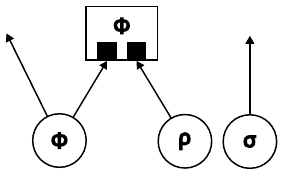}
        \end{equation*}

		\item
		All three objects are entangled. 
		This leads to a simple but lengthy calculation, which we have deferred to Appendix~\ref{subsec:verify_ent_swap}.
		The result is
		\begin{equation}
			\swap{\Phi}{\Phi}{\Phi} \subset \tfrac14 \Phi.
		\end{equation}
	\end{enumerate}

    Therefore, none of the conditions required to detect inconsistency in Algorithm~\ref{alg:induced_theory} are met.
    And hence, it does not return ``inconsistent''.
\end{proof}

\subsection{Iterated CHSH game}

The theory exhibits a CHSH value of $4$ for the observable~(\ref{eqn:chsh_obs}) and the choices
\begin{equation}\label{eqn:obl_thry_chsh_4}
    \begin{aligned}
        \rho 
        &= 
        | \Phi^+_{0,1} \rangle \langle \Phi^+_{0,1} |,
        \\
        A_0 
        & = 
        R | \tilde x_+ \rangle\langle \tilde x_+ | R^\dag
        -
        R | \tilde x_- \rangle\langle \tilde x_- | R^\dag
        =
        r R \sigma_1 R^\dag,
        \\
        A_1 
        & = 
        R | \tilde y_+ \rangle\langle \tilde y_+ | R^\dag
        -
        R | \tilde y_- \rangle\langle \tilde y_- | R^\dag
        =
        r R \sigma_2 R^\dag,
        \\
        B_0 
        & = 
        R | \tilde x_+ \rangle\langle \tilde x_+ | R^\dag
        -
        R | \tilde x_- \rangle\langle \tilde x_- | R^\dag
        =
        r R \sigma_1 R^\dag,
        \\
        B_1 
        & = 
        R | \tilde y_+ \rangle\langle \tilde y_+ | R^\dag
        -
        R | \tilde y_- \rangle\langle \tilde y_- | R^\dag
        =
        r R \sigma_2 R^\dag.
    \end{aligned}
\end{equation}
Since this is the maximum possible value, the CHSH value associated with this theory, as defined in Eqn.~(\ref{eqn:theory_chsh}), is also 4.

Also, $\forall \ \rho \in \D^{(2)}$ we have (refer to Appendix~\ref{subsec:verify_ost_stability})
    \begin{equation}
        \swap
        {
            | \Phi^+_{0,1} \rangle \langle \Phi^+_{0,1} |
        }
        {
            | \Phi^+_{0,1} \rangle \langle \Phi^+_{0,1} |
        }
        {
            \rho
        }
        \propto
        \rho
    \end{equation} 
Therefore, we not only have closure but also stability under both entanglement swapping and dual entanglement swapping. 

\begin{remark}
    It is worth pointing out here that the choice of correlators being only on the $xy$-plane is no coincidence. 
    For any situation with a $z$-measurement, oblate stabilizer theory no longer has such a strong CHSH violation.
\end{remark}

The general strategy for the iterated CHSH game for Oblate Stabilizer Theory is the same as the optimal strategy for quantum theory discussed in Sec.~\ref{sec:the_question}.

Alice and Charlie have access to two-setting two-outcome measurement machines.
The Bobs have access to four-outcome bipartite measurement machines.
In each run of the experiment, nearest neighbours share a bipartite Oblate Stabilizer State.
Each Bob performs a bipartite measurement on the sub-systems available to him and broadcasts his outcome.
Based on this, either Alice or Charlie apply a correction locally and perform a CHSH test. 

\begin{theorem}\label{thm:OST_CHSH_4}
	Oblate stabilizer theory reaches a value of $\beta=4$ in the iterated CHSH game. 
\end{theorem}

\begin{proof}
    Build the two-setting two-outcome measurement machines of Alice and Charlie out of the correlators
    \begin{equation}
        \begin{aligned}
            A_0 = rR\sigma_1R^\dag,
            \qquad 
            A_1 = rR\sigma_2R^\dag,
            \\ 
            C_0 = rR\sigma_1R^\dag, 
            \qquad
            C_1 = rR\sigma_2R^\dag.
        \end{aligned} 
    \end{equation}
    As we have noted in (\ref{eqn:obl_thry_chsh_4}), the above are valid correlators of the theory. For the four-outcome measurements, choose
    \begin{equation}
        \mathcal M_{B_i} 
        = 
        \big\{ 
        | \Phi^+_{\mu, 1} \rangle\langle \Phi^+_{\mu, 1} | 
        \ \big| \ 
        \mu \in \ZZ_4 
        \big\}, \ \forall \, i.
    \end{equation}
    From (\ref{eqn:Bell_basis}) it follows that $\mathcal M_{B_i}$ is a measurement. For the shared bipartite states, choose 
    \begin{equation}
        \rho_{AB_1} 
        = 
        \rho_{B_1B_2}
        =
        \cdots 
        =
        \rho_{B_n C}
        =
        | \Phi^+_{0, 1} \rangle\langle \Phi^+_{0, 1} |.
    \end{equation}

    Given these choices one can verify (refer to Appendix.~\ref{subsec:verify_ent_swap}) that the four output states obtained after each consecutive entanglement swap are the same.
    Each state is also equiprobable.
    The output state after $n$ rounds depends solely on the multiplicity of each of the four outcomes of $\mathcal M_{B_i}$, in the vector $\vec b = (b_1, \cdots, b_n) \in (\ZZ_4)^n$.
    Therefore to obtain the proper correction one can convert the outcomes to binary and perform a bit-wise XOR (equivalent to finding the resultant element of the Klein four-group).
    The output state and correction corresponding to each $\mu \in \ZZ_4$ is tabulated in Tab.~\ref{tab:output_and_corrections}.
    \begin{table}
        \centering
        \begin{tabular}{lcr}
            \hline 
            \hline 
            $\mu \in \ZZ_4$ \hspace*{5.75em}  & State \hspace*{5.75em} & Correction\\
            \hline 
            $0$
            \hspace*{5.75em}
            &
            $| \Phi^+_{0, 1} \rangle\langle \Phi^+_{0, 1} |$ 
            \hspace*{5.75em}
            &
            $\sigma_0 \otimes \Id$
            \\
            $1$
            \hspace*{5.75em}
            &
            $| \Phi^+_{1, -1} \rangle\langle \Phi^+_{1, -1} |$ 
            \hspace*{5.75em}
            &
            $\sigma_1 \otimes \Id$
            \\
            $2$
            \hspace*{5.75em}
            &
            $| \Phi^+_{2, -1} \rangle\langle \Phi^+_{2, -1} |$ 
            \hspace*{5.75em}
            &
            $\sigma_2 \otimes \Id$
            \\
            $3$
            \hspace*{5.75em}
            &
            $| \Phi^+_{3, 1} \rangle\langle \Phi^+_{3, 1} |$ 
            \hspace*{5.75em}
            &
            $\sigma_3 \otimes \Id$
            \\[1ex]
            \hline 
            \hline
        \end{tabular}
        \caption{The output state and correction corresponding to each outcome of Bob's measurement.}
        \label{tab:output_and_corrections}
    \end{table}

    The correction maps each output state back to $| \Phi^+_{0, 1} \rangle\langle \Phi^+_{0, 1} |$, which gives a CHSH value of $4$ with the above choice of correlators as was stated earlier.
    And therefore, we have 
    \begin{equation}
        \beta 
        = 
        \sum_{\vec b \in (\ZZ_4)^n} p_{b} \beta_{b} 
        = 
        \sum_{\vec b \in (\ZZ_4)^n} \big(\tfrac14\big)^n 4 
        = 
        4
    \end{equation}
\end{proof}


\section{Conclusion and Outlook} \label{sec:conclusion}

We have constructed a GPT in which a CHSH violation of $4$ can be sustained indefinitely under entanglement swapping.
As a consequence, the iterated CHSH game is insufficient to single out QM among GPTs.

In the process of obtaining this result, we have also set up a framework to turn bipartite theories into multipartite theories in which entanglement swapping is consistently defined. 

As an outlook to future work, Ref.~\cite{weilenmann2020toward} suggests that theories should satisfy stronger symmetry conditions, i.e.,
``...for any state and set of local measurements, if the local outcome probabilities are permuted, then there is a state that achieves these permuted correlations under the same measurements''.
If this is interpreted as an invarience under permutation of sub-systems, then oblate stabilizer theory satisfies this requirement.
If instead we interpret this as a symmetry under permutation of extremal effects, then our construction fails to satisfy this requirement. 
In particular, our construction breaks the symmetry between $z$-observables and those on the equatorial plane.
This raises two complementary questions for further work:
(1) Are there natural, stronger conditions on multipartite correlations for which QM is indeed optimal?
(2) Can one find a theory that beats QM in the iterated CHSH game and that is \emph{isotropic} in the sense of having a transitive symmetry group action on all extremal effects?


\section{Acknowledgements}

We thank 
Roger Colbeck,
Marc-Olivier Renou,
Mirjam Weilenmann, 
and
Elie Wolfe, 
for their discussions on composite GPTs, and Markus P.\ M\"{u}ller for his insight on generating examples of GPTs.
We also thank Ludovico Lami for his discussion on the history of GPTs.

This work has been supported by Germany's Excellence Strategy --
Cluster of Excellence Matter and Light for Quantum Computing (ML4Q) EXC 2004/1
- 390534769 and the German Research Council (DFG) via contract GR4334/2-2.


\bibliography{bibliography}


\onecolumngrid

\section{Appendix}

\noindent The Bell state can be written concisely as
\begin{equation}
    |\Phi^+ \rangle \langle \Phi^+ |
    =
    \tfrac14\sum_{\mu = 0}^3 (-1)^{\delta_{2\mu}} \sigma_\mu \otimes \sigma_\mu
\end{equation}

\subsection{Stability of OST}\label{subsec:verify_ost_stability}

To verify that oblate stabilizer theory is stable under entanglement swapping we show that 
\begin{equation}
    \swap
        {
            | \Phi^+_{0,1} \rangle \langle \Phi^+_{0,1} |
        }
        {
            | \Phi^+_{0,1} \rangle \langle \Phi^+_{0,1} |
        }
        {
            \bullet
        }
    :
    \sigma_\mu \otimes \sigma_\nu 
    \longmapsto 
    \tfrac14 \sigma_\mu \otimes \sigma_\nu
\end{equation}

To this end, note that 
\begin{equation}
    \begin{aligned}
        &\swap
        {
            | \Phi^+_{0,1} \rangle \langle \Phi^+_{0,1} |
        }
        {
            | \Phi^+_{0,1} \rangle \langle \Phi^+_{0,1} |
        }
        {
            \sigma_\mu \otimes \sigma_\nu
        }
        \\
        &=
        \tr_{23}\big(
        (R^\dag \otimes \Id)_{12}
        |\Phi^+ \rangle \langle \Phi^+ |_{12} 
        (R \otimes \Id)_{12}
        (R^\dag \otimes \Id)_{23}
        |\Phi^+ \rangle \langle \Phi^+ |_{23}
        (R \otimes \Id)_{23}
        (\sigma_\mu \otimes \sigma_\nu)_{34}
        \big)
        \\
        &=
        \tr_{23}\big(
        (\Id \otimes R^{-1})_{12}
        |\Phi^+ \rangle \langle \Phi^+ |_{12} 
        (R^{-1}R \otimes \Id^{\otimes 2})_{123}
        |\Phi^+ \rangle \langle \Phi^+ |_{23}
        (R \otimes \Id)_{23}
        (\sigma_\mu \otimes \sigma_\nu)_{34}
        \big)
        \\
        &=
        \tr_3\Big(
        \tfrac1{16}    
        \sum_{\alpha,\beta = 0}^3 (-1)^{\delta_{2\alpha} + \delta_{2 \beta}}
        \tr(R^{-1}\sigma_\alpha \sigma_\beta R) 
        (\sigma_\alpha \otimes \sigma_\beta)_{13}
        (\sigma_\mu \otimes \sigma_\nu)_{34}
        \Big)
        \\
        &=
        \tfrac18    
        \sum_{\alpha= 0}^3 
        \tr(\sigma_\alpha \sigma_\mu)
        (\sigma_\alpha \otimes \sigma_\nu)
        =
        \tfrac14 \sigma_\mu \otimes \sigma_\nu
    \end{aligned}
\end{equation}

Every $\rho \in \D^{(2)}$ is a linear combination of $\sigma_\mu \otimes \sigma_\nu$ for $\mu, \nu \in \ZZ_4$, hence the claim follows.

\subsection{Iterated Entanglement Swapping}\label{subsec:verify_ent_swap}

\noindent In order to find the result of $\swap{\Phi}{\Phi}{\Phi}$, let us first calculate the following identity: for some operators $A, B, C, D$
\begin{equation}\label{eqn:appendix_id_1}
    \begin{aligned}
        & \tr_{23}\big(
        (A \otimes \Id)_{12}
        |\Phi^+ \rangle \langle \Phi^+ |_{12} 
        (B \otimes \Id^{\otimes 2})_{123}
        |\Phi^+ \rangle \langle \Phi^+ |_{23}
        (\Id^{\otimes 2} \otimes C)_{234}
        |\Phi^+ \rangle \langle \Phi^+ |_{34}
        (\Id \otimes D)_{34}
        \big)
        \\
        &=
        \tfrac1{64}
        \sum_{\mu,\nu,\lambda = 0}^3 \Big( (-1)^{\delta_{2\mu} + \delta_{2\nu}+ \delta_{2\lambda}}
        \tr(\sigma_\mu \sigma_\nu) 
        \tr(\sigma_\nu \sigma_\lambda)
        A \sigma_\mu B \otimes C \sigma_\lambda D
        \Big)
        \\
        &=
        \tfrac1{16} (A \otimes C)\Big(\sum_{\mu = 0}^3 (-1)^{\delta_{2\mu}} \sigma_\mu \otimes \sigma_\mu \Big) (B \otimes D)
        \\
        &=
        \tfrac14 (A \otimes C) |\Phi^+ \rangle \langle \Phi^+ | (B \otimes D)
        =
        \tfrac14 (AC^t \otimes \Id) |\Phi^+ \rangle \langle \Phi^+ | (D^tB \otimes \Id).
    \end{aligned}
\end{equation}

\noindent Using this we can now compute 
\begin{equation}\label{eqn:appendix_id_2}
    \begin{aligned}
        & \tr_{23}\big(
        (A^\dag \otimes \Id)_{12}
        |\Phi^+ \rangle \langle \Phi^+ |_{12} 
        (A \otimes \Id)_{12}
        (B^\dag \otimes \Id)_{23}
        |\Phi^+ \rangle \langle \Phi^+ |_{23}
        (B \otimes \Id)_{23}
        (C^\dag \otimes \Id)_{34}
        |\Phi^+ \rangle \langle \Phi^+ |_{34}
        (C \otimes \Id)_{34}
        \big)
        \\
        &= \tr_{23}\big(
        (A^\dag \otimes \Id)_{12}
        |\Phi^+ \rangle \langle \Phi^+ |_{12} 
        (\overline{B}A \otimes \Id^{\otimes 2})_{123} 
        |\Phi^+ \rangle \langle \Phi^+ |_{23}
        (\Id^{\otimes 2} \otimes \overline{C}B)_{234}
        |\Phi^+ \rangle \langle \Phi^+ |_{34}
        (\Id \otimes C^t)_{34}
        \big)
        \\
        &=
        \tfrac14 (A^\dag (\overline{C}B)^t \otimes \Id) |\Phi^+ \rangle \langle \Phi^+ | ((C^t)^t\overline{B}A \otimes \Id)
        \\
        &=
        \tfrac14 (C \overline{B}A)^\dag \otimes \Id |\Phi^+ \rangle \langle \Phi^+ | (C \overline{B}A) \otimes \Id.
    \end{aligned}
\end{equation}

\noindent Therefore, the entanglement swapping map
\begin{equation}
    \begin{aligned}
        \swap
        {
            |\Phi^+_{\mu, m} \rangle \langle \Phi^+_{\mu, m} |
        }{
            |\Phi^+_{\nu, m'} \rangle \langle \Phi^+_{\nu, m'} |
        }{
            |\Phi^+_{\lambda, m''} \rangle \langle \Phi^+_{\lambda, m''} |
        }
        &
        =
        \tr_{23}\big(
        ((\sigma_\mu R^m)^\dag \otimes \Id)_{12}
        |\Phi^+ \rangle \langle \Phi^+ |_{12} 
        ((\sigma_\mu R^m) \otimes \Id)_{12}
        \\
        & \qquad \qquad
        ((\sigma_\nu R^{m'})^\dag \otimes \Id)_{23}
        |\Phi^+ \rangle \langle \Phi^+ |_{23}
        ((\sigma_\nu R^{m'}) \otimes \Id)_{23}
        \\
        & \qquad \qquad
        ((\sigma_\lambda R^{m''})^\dag \otimes \Id)_{34}
        |\Phi^+ \rangle \langle \Phi^+ |_{34}
        ((\sigma_\lambda R^{m''}) \otimes \Id)_{34}
        \big)
    \end{aligned}
\end{equation}
reduces to~(\ref{eqn:appendix_id_2}) with 
\begin{equation*}
    A = \sigma_\mu R^m, 
    \qquad
    B = \sigma_\nu R^{m'},
    \qquad
    C = \sigma_\lambda R^{m''}. 
\end{equation*}
which gives
\begin{equation}
    C\overline{B}A = \sigma_\lambda R^{m''} \overline{\sigma_\nu R^{m'}} \sigma_\mu R^{m} = (\pm i)\sigma_\xi R^{m^*},
\end{equation}

\noindent with $\xi \in \ZZ_4$ and $m^* \in \ZZ_8$ and odd.
Therefore, it follows that 
\begin{equation}
    \swap{\Phi}{\Phi}{\Phi}
    = \tfrac14 \Phi.
\end{equation}

\noindent Finally, we can specialize to the case of the optimal strategy for the iterated CHSH game. 
For $n=1$:
\begin{equation*}
    A = \sigma_0 R,
    \quad
    B = \sigma_\mu R,
    \quad
    C = \sigma_0 R
    \implies 
    C\overline{B}A = R \sigma_\mu
\end{equation*}

\noindent We can calculate the output of $n = 2$ by
entering the output of the first round into the second round, i.e., 
\begin{equation*}
    A = R \sigma_\mu,
    \quad
    B = \sigma_\nu R,
    \quad
    C = \sigma_0 R 
    \implies 
    C\overline{B}A 
    = \sigma_0 R \overline{\sigma_\nu R} R \sigma_\mu
    = R \sigma_\nu \sigma_\mu.
\end{equation*}

\noindent These are the same four states as $n=1$ up to scaling and factors of $\pm i$, which are eliminated since they come in complex conjugate pairs.
This means that the set of output states is closed.
Moreover, each time we have conjugation by an additional Pauli matrix.
The final result therefore depends only on the number of time each Pauli matrix occurs.
That is, it depends on the multiplicity of each member of $\ZZ_4$ in the outcome vector $\vec b = (b_1, \cdots, b_n)$.

\end{document}